\documentclass[journal,compsoc,twoside]{IEEEtran}

\usepackage[dvipsnames]{xcolor}

\usepackage{mathtools}
\usepackage{booktabs}

\ifCLASSOPTIONcompsoc
    \usepackage[caption=false, font=normalsize, labelfont=sf, textfont=sf]{subfig}
\else
\usepackage[caption=false, font=footnotesize]{subfig}
\fi

\usepackage{pgfplots}
\usepgfplotslibrary{groupplots}

\usepackage{siunitx}
\usepackage{amsmath}
\usepackage{amsthm}
\usepackage{amsfonts}
\usepackage{textcomp}

\usepackage{dashbox}

\PassOptionsToPackage{hyphens}{url}
\usepackage[hidelinks]{hyperref}
\usepackage[hyphenbreaks]{breakurl}

\usepackage{algorithm}
\usepackage[algosection,vlined,algo2e,linesnumbered,algoruled]{algorithm2e}

\newcommand{\fl}[1]{\ensuremath{\mathrm{fl}(#1)}}
\newcommand{\flp}[2]{\ensuremath{\mathrm{fl}_{#1}}(#2)}
\newcommand{\flr}[1]{\ensuremath{\mathrm{flr}(#1)}}

\newcommand{\ulp}[1]{\ensuremath{\mathrm{ulp}(#1)}}

\newtheorem{theorem}{Theorem}[section]
\newtheorem{corollary}{Corollary}[theorem]

\IEEEtitleabstractindextext{
  \begin{abstract}
      In the literature on algorithms for computing multi-term addition $s_n=\sum_{i=1}^n x_i$ in floating-point arithmetic it is often shown that a hardware unit that has single normalization and rounding improves precision, area, latency, and power consumption, compared with the use of standard add or fused multiply--add units.
  However, non-monotonicity can appear when computing sums with a subclass of multi-term addition units, which is currently not explored in the literature.
  We prove that computing multi-term floating-point addition with $n\geq 4$, without normalization of intermediate quantities, can result in non-monotonicity---increasing one of the addends $x_i$ decreases the sum $s_n$.
  Summation is required in dot product and matrix multiplication operations, operations that are increasingly appearing in the hardware of high-performance computers, and knowing where monotonicity is preserved can be of interest to the developers and users.
  Non-monotonicity of summation in existent hardware devices that implement a specific class of multi-term adders may have appeared unintentionally as a consequence of design choices that reduce circuit area and other metrics.
  To demonstrate our findings we simulate non-monotonic multi-term adders in MATLAB using the \texttt{CPFloat} custom-precision floating-point simulator.
\end{abstract}
}

\title{Monotonicity of Multi-Term Floating-Point Adders}
\author{Mantas Mikaitis\thanks{%
		School of Computing,
		University of Leeds,
		Leeds, LS2 9JT, United Kingdom
		(\texttt{M.Mikaitis@leeds.ac.uk}).
    Preprint version: 4th of December, 2023.}}

\begin{document}

\maketitle

\section{Introduction}

A real function $f$ is monotonically nondecreasing on an interval $[a,b]$ if $f(x) \leq f(y)$ whenever $a \leq x \leq y \leq b$.
In other words, when the argument of a monotonic function is increasing, the value of the function does not decrease.
Similarly for a function that is monotonic nonincreasing: when the input argument is increasing the value of the function is not increasing.
For the multivariate functions, a function is monotonic if it is monotonic for all the input values.
For example, $f(x_1, x_2)$ is monotonic nondecreasing if for any $x_1 \leq x^*_1$ and $x_2 \leq x^*_2$ we have $f(x_1, x_2) \leq f(x^*_1, x^*_2)$.

Summation of a set of values is a multivariate function that is monotonic nondecreasing (just monotonic thereafter).
Given a set of input values $x_1$, $x_2$, $...$, $x_n$, summation is expressed as
\begin{equation}
  f(x_1, x_2, ..., x_n) = \sum_{i=1}^n x_i.
  \label{eq:sum}
\end{equation}
This function is monotonically nondecreasing by the definition of the sum: take $x_n = a$ followed by $x_n = a + \varepsilon$, with $\varepsilon > 0$.
Then $f(x_1, x_2, ..., a) = (\sum_{i=1}^{n-1} x_{i}) + a$, whereas $f(x_1, x_2, ..., a + \varepsilon) = (\sum_{i=1}^{n-1} x_{i}) + a + \varepsilon > (\sum_{i=1}^{n-1} x_{i}) + a$.
Because of the commutativity of the sum this is true for all the input arguments.

Summation of values is at the core of scientific computing---it is required, for example, for calculating vector--vector products, matrix--vector and matrix--matrix multiplications, as well as in evaluating polynomials.
Most computer software works with floating-point numbers \cite{ieee19}, rather than exact numbers, so the addition operation is different from the exact addition and the monotonicity of summation should be tested rather than assumed to hold because it holds in exact arithmetic.
This is similar to the properties of associativity, commutativity and distributivity---commutativity is generally preserved, but associativity and distributivity in basic arithmetic operations, when transitioning from exact to floating-point arithmetic, are not~\cite[Sec.~2.6]{mbdj18}.

From here we use the notation of the standard model \cite[Sec.~2.2]{high:ASNA2} of addition in precision-$p$ arithmetic:
$$
\fl{x+y}=(x+y)(1+\delta), \quad |\delta| \leq u
$$
where $u=2^{-p}$, the unit roundoff, and $\fl{x}$ refers to normalizing (see Section~\ref{sec:fp-arith}) and rounding $x$ to form a floating-point value defined above.

Most floating-point units on general-purpose hardware include 2-term adders that compute the sum of floating-point numbers $z=\fl{a+b}$ which is required by the IEEE 754 standard \cite[Sec.~5.4.1]{ieee19}.
This operation includes computing $a+b$ as though in infinite precision, normalizing the resultant significand if required and rounding it to obtain $z$ \cite[Sec.~7.3]{mbdj18}.
At software level, this operation is called repeatedly to compute sums of arbitrary length (Equation~\ref{eq:sum}).
A high-level algorithm \cite[Sec.~4.2]{high:ASNA2} is given as Algorithm~\ref{alg:software-sum}.
On line 3 the operation $\fl{a+b}$ includes rounding and normalization, which means that, when implemented this way, overall the $n$-term addition performs $n-1$ such roundings and normalizations.
It is safe to say that most software is implemented this way on the hardware that includes circuitry only for adding two numbers as per IEEE 754.

\begin{algorithm2e}[t]
  \caption{Given numbers $\mathcal{S} = \{x_1, ..., x_n\}$, compute $s_n=\sum_{i=1}^{n}x_i$.}
  \label{alg:software-sum}
  Repeat while $\mathcal{S}$ contains more than one element\\
    \Indp From $\mathcal{S}$, remove two numbers $a$ and $b$\\
    Put $\fl{a+b}$ to $\mathcal{S}$\\
  \Indm
  The remaining element in $\mathcal{S}$ is $s_n$
\end{algorithm2e}

Summation of three, four or longer vectors of floating-point values can be classified under reduction operations which are recommended but not required by the IEEE 754 standard~\cite[Sec.~9.4]{ieee19}.
The guideline provided by IEEE 754 for implementing them is not significantly constrained.

Addition of more than three values can be achieved by an accumulator which takes as one of the inputs its output from the previous stage or by using multiple 2-term adders in series or in parallel (Figure~\ref{fig:adders}).
However, in hardware, a custom design is often considered to gain in speed and circuit area compared with the straightforward approach of using the standard 2-term adders.
Algorithm~\ref{alg:multi-operand-sum} shows the main steps taken by most of the multi-term adders available in the literature (see Section~\ref{sec:fp-arith} for the details on floating point).
The main point to note is that normalization and rounding are performed once, at the end, rather than after each intermediate addition operation, as in Algorithm~\ref{alg:software-sum}, which in literature on hardware appears to be beneficial for saving hardware resources and even increasing the accuracy.
All of this applies to dot products and matrix multiplies, by replacing $\fl{a+b}$ with $\fl{\fl{a\times b}+c}$ or the fused multiply--add (FMA) $\fl{a\times b+c}$.

\begin{algorithm2e}
  \caption{Given numbers $\mathcal{S} = \{x_1, ..., x_n\}$, with exponents $\{e_1, ..., e_n\}$ and significands $\{1.m_1, ..., 1.m_n\}$ compute $s_n=\sum_{i=1}^{n}x_i$.}
  \label{alg:multi-operand-sum}
  Determine $e_{max}=\mathrm{max}(e_1,...,e_n)$\\
  Align all $1.m_i$ by shifting each $e_{max}-e_i$ steps right\\
  Perform addition of aligned significands\\
  Perform normalization and rounding to form $s_n$\\
\end{algorithm2e}

As a side note, unnormalized floating-point arithmetic appeared as early as 1958 in the work of Metropolis~and~Ashenhurst~\cite{meas58}; the performance improvement compared with the normalized arithmetic was noted even then.

In this paper we identify four types of implementation for the Algorithm~\ref{alg:multi-operand-sum} and demonstrate that one class of designs for multi-term addition are non-monotonic, including various commercial hardware designs available and widely present on the machines in the TOP500\footnote{\url{https://top500.org/lists/top500/list/2023/11/}}.
We also suggest that this issue can be fixed by masking off the bottom bits when carries occur or by using an adder from a different class, which can be added into the future summation and dot product hardware designs as an option.

We mentioned that non-monotonic summation in floating-point arithmetic adds to the list of mathematical properties that are not preserved when switching from exact arithmetic.
One other motivating point for studying this is reproducibility of numerical computations.
Bit-wise reproducibility is not impossible on single-core CPUs that implement the IEEE 754 standard correctly and assuming special features such as 80-bit arithmetic are not enabled by compilers.
If we run some code in IEEE 754 software, using basic operations of a floating-point unit (FPU), we get one behaviour, but if we run that code in hardware that performs summation non-monotonically, we may get unexpected results that may be hard to explain.
We now demonstrate this with an example.

\subsection{An example on current hardware}

For the following we use NVIDIA A100 SXM 80GB GPU because we have access to it, however we predict that most commercial hardware in the market would not pass the following, and similar, tests.

A100 GPUs are equipped with matrix multiply hardware that can perform $D = A\times B + C$, where $A\in \mathbb{R}^{8\times 8}$ and $B \in \mathbb{R}^{8\times 4}$ are binary16 matrices, and $C,D \in \mathbb{R}^{8\times 4}$ are binary32 matrices \cite[p.~20]{nvid20a}.
We showed before that matrix multipliers in this GPU have one extra bit in the alignment of binary32 significands \cite{fhmp21}.
We will focus on two resulting matrix elements, which perform dot products $d_{11}=a_{11}b_{11} + a_{12}b_{21} + \cdots + a_{18}b_{81} + c_{11} $ and $d_{12}=a_{11}b_{12} + a_{12}b_{22} + \cdots + a_{18}b_{82} + c_{12} $ by most likely implementing a 9-term adder of products.
We set $A,B=1$ (matrices of ones) and $c_{11}=33554430$ and $c_{12}=33554432$.
Computing $A\times B+C$ with the GPU matrix multipliers returns a matrix that has $d_{11}=33554436$ and $d_{12}=33554432$, demonstrating non-monotonic behaviour where an increase in one of the 9 addends decreases the sum.
Note that we increased the addend by 2 and the sum was decreased by 4 since the amount we decrease by comes from the other 8 addends not contributing to the sum.
Similar examples can be constructed around any powers of two, and at the edges of the dynamic range the quantities by which the sum changes by decreasing one of the addends would be larger in absolute terms.
In the past we performed tests that indicated that NVIDIA V100 and the T4 GPUs have similar behaviour, although the V100 differs due to narrower internal accumulator in the multi-term addition \cite{fhmp21}.
NVIDIA H100 and the AMD matrix multipliers are yet to be tested.

\subsection{Multi-term floating-point addition in literature}
\label{sec:classes}

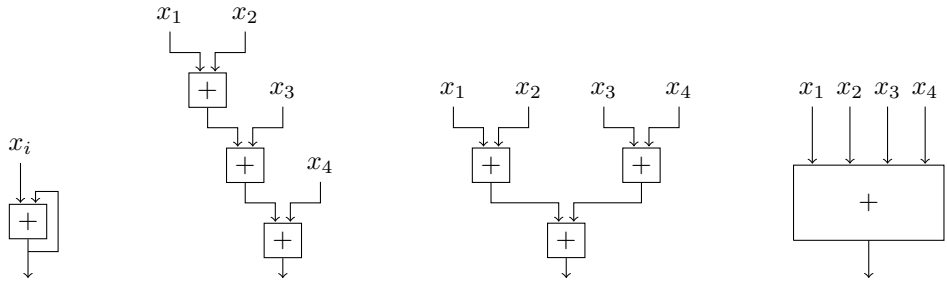
\begin{figure*}
  \centering
  \subfloat{
      \begin{tikzpicture}
        \node[rectangle,draw] (add1) at (0,0) {$+$};
        \node (in1) at (-0.1, 1) {$x_i$};
        
        \draw[->, black] (in1) -- (-0.1, 0.25);
        \draw[->, black] (add1) -- (0, -0.75);
        \draw (0, -0.4) -- (0.4, -0.4)
        -- (0.4, 0.4) -- (0.1, 0.4)
        [->, black] -- (0.1, 0.25);
        
      \end{tikzpicture}}
    \hspace{1cm}
    \subfloat{
      \begin{tikzpicture}
        \node[rectangle,draw] (add1) at (0,0) {$+$};
        \node (in1) at (-0.5, 1) {$x_1$};
        \node (in2) at (0.5, 1) {$x_2$};
        
        \draw (in1) -- (-0.5, 0.5)
        -- (-0.1, 0.5)
        [->,black] -- (-0.1, 0.25);

        \draw (in2) -- (0.5, 0.5)
        -- (0.1, 0.5)
        [->,black] -- (0.1, 0.25);
        
        \node[rectangle,draw] (add2) at (0.5,-1) {$+$};
        \node (in3) at (1, 0) {$x_3$};
        
        \draw (add1) -- (0, -0.5)
        -- (0.4, -0.5)
        [->,black] -- (0.4, -0.75);

        \draw (in3) -- (1, -0.5)
        -- (0.6, -0.5)
        [->,black] -- (0.6, -0.75);

        \node[rectangle,draw] (add3) at (1,-2) {$+$};
        \node (in4) at (1.5, -1) {$x_4$};
        
        \draw (add2) -- (0.5, -1.5)
        -- (0.9, -1.5)
        [->,black] -- (0.9, -1.75);

        \draw (in4) -- (1.5, -1.5)
        -- (1.1, -1.5)
        [->,black] -- (1.1, -1.75);

        \draw[->] (add3) -- (1, -2.5);

      \end{tikzpicture}}
    \hspace{1cm}
    \subfloat{
      \begin{tikzpicture}
        \node[rectangle,draw] (add1) at (0,0) {$+$};
        \node (in1) at (-0.5, 1) {$x_1$};
        \node (in2) at (0.5, 1) {$x_2$};

        \draw (in1) -- (-0.5, 0.5)
        -- (-0.1, 0.5)
        [->,black] -- (-0.1, 0.25);

        \draw (in2) -- (0.5, 0.5)
        -- (0.1, 0.5)
        [->,black] -- (0.1, 0.25);

        \node[rectangle,draw] (add2) at (2,0) {$+$};
        \node (in3) at (1.5, 1) {$x_3$};
        \node (in4) at (2.5, 1) {$x_4$};

        \draw (in3) -- (1.5, 0.5)
        -- (1.9, 0.5)
        [->,black] -- (1.9, 0.25);

        \draw (in4) -- (2.5, 0.5)
        -- (2.1, 0.5)
        [->,black] -- (2.1, 0.25);

        \node[rectangle,draw] (add3) at (1,-1) {$+$};

        \draw (add2) -- (2, -0.5)
        -- (1.1, -0.5)
        [->,black] -- (1.1, -0.75);

        \draw (add1) -- (0, -0.5)
        -- (0.9, -0.5)
        [->,black] -- (0.9, -0.75);

        \draw[->,black] (add3) -- (1, -1.5);
        
      \end{tikzpicture}}
    \hspace{1cm}
    \subfloat{
      \begin{tikzpicture}
        \draw (0,0) -- (2,0) -- (2,1) -- (0,1) -- (0,0);
        \node (text) at (1,0.5) {$+$};
        \node (in1) at (0.25, 2) {$x_1$};
        \node (in2) at (0.75, 2) {$x_2$};
        \node (in3) at (1.25, 2) {$x_3$};
        \node (in4) at (1.75, 2) {$x_4$};
        \draw[->,black] (in1) -- (0.25, 1.02);
        \draw[->,black] (in2) -- (0.75, 1.02);
        \draw[->,black] (in3) -- (1.25, 1.02);
        \draw[->,black] (in4) -- (1.75, 1.02);
        \draw[->,black] (1, 0) -- (1, -0.5);
      \end{tikzpicture}}
    \caption{Various methods for implementing 4-term addition in hardware. The left-most method, which is present in virtually all modern hardware, is to iterate over an IEEE 754 floating-point addition multiple times to sum a vector of numbers $x$; with this method, every iteration requires a new addition instruction to be fetched and executed. The second and third approaches demonstrate two ways to utilize multiple IEEE 754 floating-point adders---the first one simply chains four adders while the second one creates an adder tree to reduce the latency. In these three methods, the order of addition can be changed by rearranging the inputs. The remaining method is a black box which contains a specialized hardware design for performing the summation, not necessarily using the operations outlined in IEEE 754 and usually enforcing a specific order of summation. The first three approaches correspond to Class III adders and the last one encapsulates Classes I, II, and IV (Section~\ref{sec:classes}).}
    \label{fig:adders}
  \end{figure*}

  Hardware designs of multi-term adders (Algorithm~\ref{alg:multi-operand-sum}), including those that are part of dot product and matrix multiply hardware, can be classified into four main categories.
  Some of the ways to build multi-term adders are demonstrated in the high-level diagrams of Figure~\ref{fig:adders}.

  We will use a term \emph{fused}. From the user's perspective, in most cases it means that only one rounding error is incurred in the computation, except where stated otherwise.

  \subsubsection{Class I: Adders that use long accumulators}
  \label{sec:class1}

One approach is to retain all the bits in the summation of multiple values and round it once at the end.
This is advocated by Kulisch~\cite[Sec.~8]{kuli13}.
See the design-space exploration by Uguen~and~de~Dinechin~\cite{ugdi17} for a detailed analysis of the costs.
An implementation by Koenig,~Bachrach,~and Asanovi\'{c}~\cite{kbba17} used 4288 bits internally for multiplying and accumulating binary64 values exactly.
These kind of multi-term adders are fused because they contain only one rounding across the whole computation; however, keeping all of the bits can be expensive in circuit area and latency due to carry propagation.

While not directly a multi-term adder, the accumulator of binary16 products by Brunie~\cite{brun17} uses an exact 80-bit fixed-point internal format and therefore can be used to implement fused dot products or matrix multiplies with a single rounding error.
Brunie~\cite{brun20} also proposed an architectural extension to CPUs which adds basic linear algebra instructions that work on matrices packed in general-purpose vector registers.
The papers suggest that whether the accumulation is exact or not depends on the precision of input arguments and whether it is feasible to build the hardware required to accumulate exactly.

Burgess,~Goodyer,~Hinds,~and~Lutz~\cite{bghl19} propose High-Precision Anchored (HPA) accumulators for accurate floating-point summation and suggest extensions to ARM Scalable Vector Extension (SVE) units to efficiently support them.
The HPA number format would be used for computation, while the input and output data would still be in an IEEE 754 format, such as the binary64, which requires conversion to and from the HPA.
The main concept is to convert a floating-point value into the HPA format by placing different parts of the significand into different registers, based on the significance of the bits when taking into account the exponent.
HPA numbers are therefore stored in a wider format, across multiple registers.
Similar to fixed-point representation, scaling is applied to choose the balance between the range and precision.
HPA numbers can be configured to calculate correctly rounded sums of floating-point values, therefore we classify this software-hardware concept in the Class I of multi-term adders.

\subsubsection{Class II: Adders that achieve correct rounding without the use of long accumulators}
\label{sec:class2}

Tenca~\cite{tenc09} provides an optimized algorithm for finding the largest exponent and choosing the amounts to shift the significands by.
The general algorithm is not changed, however, with the main steps in Algorithm~\ref{alg:multi-operand-sum} still present.
Tenca~\cite{tenc09} actually proposes a fused design for performing $\fl{a+b+c}$ with only one rounding error, which complicates the problem in that bits that are shifted out in the significand alignment step have to be tracked.
This is not what is implemented in the hardware; for example in the A100, which we used for the demonstration above, $n-1$ rounding or truncation errors are incurred when aligning and adding significands in limited precision.

Sohn~and~Swartzlander propose a series of fused operators, such as a two-term dot product~\cite{sosw13}, a three-term adder~\cite{sosw14}, and a four-term dot product~\cite{sosw16}.
A generalized $n$-term fused dot product architecture is explored by Tao~et~al.~\cite{tdxn13}.
The goal is to implement fused operations, meaning that computation is not performed through standard hardware multipliers and adders joined together, but by making a new optimized unit without the intermediate rounding and normalization steps.
Since these operators are fused, we do not expect non-monotonicity to appear when computing with them.

Multi-term adders also appear is in the hardware designs of the fast Fourier transform (FFT) operation.
Swartzlander~and~Saleh~\cite{swsa12} utilize a two-term adder for implementing a fused two-term dot product while Kaivani~and~Ko~\cite{kako15} discuss an implementation of FFT for which a five-term floating-point adder was used.
The authors mention not using intermediate normalization and rounding blocks by implementing a custom-design five-term adder, which in turn allowed to reduce the area of the FFT design.
This 5-term adder is fused, but it is not specified what the accumulator's size is and how the sticky bits are computed to replicate the exact accumulation.

A generalized algorithm by Boldo,~Gallois-Wong,~and~Hillaire~\cite{bgh20} computes a correctly rounded dot product of a series of fixed-point numbers with varied precisions.
Instead of using a long accumulator that could cover all possible values, the algorithm uses some number of extra bits and round-odd rounding mode.

\subsubsection{Class III: Adders that replicate software behaviour}
\label{sec:class3}

Kim~and~Kim~\cite{kiki09} propose a 4-term dot product unit without the intermediate normalization of sums but with intermediate rounding performed in correct places (by taking into account where the most significant nonzero bit is) to assure bit reproducible operation compared with an IEEE 754 software implementation.
Even though the monotonicity is not addressed in this work, the implementation should be monotonic as it mimics a software implementation with the correctly rounded elementary operations.
The application space is 3D graphics---for this a 4-term dot product in single precision is particularly useful and this is what the authors explored.
It was noted that ``\emph{the exact bit-level matching between hardware units and software models is more important in 3D graphics than the rounding errors to the real value.}''~\cite[p.~892]{kiki09} as the motivation for performing rounding in the intermediate calculations.

\subsubsection{Class IV: Adders that use limited precision accumulator}
\label{sec:class4}

In machine learning, Kaul~et~al.~\cite{kamk19} discuss a generalized $n$-term dot product hardware design.
It is proposed to split the calculation of the maximum exponent and the differences to all the other exponents into two phases, to reduce the critical path of the design.
For the purposes of our study, the main feature of Kaul~et~al.~\cite{kamk19} design is that it implements the behaviour of Algorithm~\ref{alg:multi-operand-sum}: it aligns the products relative to the product with the maximum exponent, the alignment right-shifts are in limited precision and the addition is performed with extra bits for carries, with a single normalization step of the sum.

Lopes~and~Constantinides~\cite{loco10} have designed a configurable dot product unit which was tested on FPGAs for up to 150 terms and compared it with a basic implementation that uses a tree of multipliers and adders.
The main feature of the design is that internally it uses a configurable precision fixed-point register to accumulate the products in before normalizing and rounding it to produce the floating-point answer.
This design follows the general structure of Algorithm~\ref{alg:multi-operand-sum} with precision growth due to no intermediate normalization (as the authors point out, precision growth allows to avoid overflows), and therefore should be non-monotonic.

Hickmann~et~al.~\cite{hcry20} present a $32\times 32$ matrix multiply accelerator with 16-bit floating-point inputs and 32-bit outputs.
The internal accumulation of products is limited to 37 bits.
One interesting aspect of this work is that the term \emph{fused} is used, but since only the products are exact, not the accumulation of them, this has a different meaning than the Class I/II designs which perform everything as though the computation is exact and rounded once.
Another aspect worth noting is that this design explicitly adds the products $a_ib_i$ of the $32\times 32$ dot product operation before adding their sum to the accumulator $c_i$ in order to reduce the error accumulation when the value in the accumulator is growing in magnitude.

Bertaccini~et~al.~\cite{bpfm22} developed a three-term addition unit for performing dot products or sums of floating-point values of various formats.
The internal accumulator is expanded twice, after each addition, but it is not reported to be an exact accumulator that would cover all right-shift distances in the alignment of addends.

Lee~et~al.~\cite{lasz21} implemented a four-core mixed-precision AI chip which includes a three-term floating-point adder as part of the FMMA (fused multiply--multiply--accumulate) instruction.
The multiplicands are 8-bit floating-point values and the accumulator is 16-bit.
The three-term addition is performed in 16-bit precision followed by normalization and rounding.

\subsubsection{Adders that lie in multiple classes}

Ledoux~and~Casas~\cite{luca22} proposed a hardware generator of general matrix multiply--accumulate (GEMM) accelerators.
The generator is parameterized and provides a choice of numerical formats and the option for setting the sizes of internal accumulators, including making them long accumulators to accumulate products exactly.
Irrespective of the setup of the accumulator, rounding and normalization from the internal hardware numerical format to some chosen standard format is performed at the end, once the whole dot product has been computed.
The work does not mention the use of multi-term adders and implements GEMM accelerators through the accumulation of values by iterating through the hardware.
Nevertheless, in terms of resultant numerical behaviour, this work potentially can generate hardware of classes I and IV listed above.

\subsection{Commercial hardware}

\begin{table*}
  \centering
  \caption{List of devices that contain vector or matrix arithmetic hardware, such as dot product and matrix multiply. In the last column we make a prediction on the class of the multi-term addition based on the available information. $^\dagger$This number is determined only from the H100 whitepaper as no other information is available, to the best of our knowledge, on what inputs tensor cores take at hardware layer; $^*$Two 8-bit floating-point formats are available, one with a 4-bit exponent and a 3-bit significand, and one with a 5-bit exponent and a 2-bit significand; $^\ddagger$Configurable Floating Point 8-bit data type, with programmable bias.}
  \begin{tabular}{lrp{2cm}cccr}
    \toprule
    Year & Device/Architecture & Input formats & Output formats & Multi-term adder terms & Throughput (max) & Predicted class \\
    \midrule
    2016 & Google TPUv2 \cite{jhag21} & bfloat16 & binary32 & - & 46 Tflop/s & Class III \\
    2017 & Google TPUv3 \cite{jhag21} & bfloat16 & binary32 & - & 123 Tflop/s & Class III \\
    2018 & NVIDIA V100 & binary16 & binary32 & 5 & 125 Tflops/s & Class IV \\
    2018 & Graphcore IPU1 & binary16 & binary32 & - & 125 Tflop/s & - \\
    2020 & Google TPUv4i \cite{jhag21} & bfloat16 & binary32 & 4 & 138 Tflop/s & Class IV \\
    2020 & Graphcore IPU2 & binary16 & binary32 & - & 250 Tflop/s & - \\
    2020 & NVIDIA A100 \cite{nvid20a} & bfloat16,\hspace{1cm} binary16, binary64, TensorFloat-32 & binary32/64 & 9 & 312 Tflop/s & Class IV \\
    2021 & AMD MI250X \cite{amd21a} & bfloat16,\hspace{1cm} binary16, binary32, binary64 & - & 5 & 383 Tflop/s & - \\
    2021 & GroqChip \cite{aklk22} & binary16 & binary32 & 160 & 188 Tflops/s & Class I or II \\
    2022 & NVIDIA H100 & 8-bit$^*$, \hspace{1cm} bfloat16,\hspace{1cm} binary16, binary64, TensorFloat-32 & binary32, binary64 & 17$^{\dagger}$ & 1978.9 Tflop/s & -\\
    2022 & Intel Ponte Vecchio \cite{gksi22} & bfloat16,\hspace{1cm} binary16, binary64, TensorFloat-32 & - & - & - & - \\
    2016-2022 & Intel AMX \cite{intl22} & binary16 & binary32 & 17 & -  & Class III \\
    2023 & Tesla Dojo \cite{tswa23} & CFP8$^\ddagger$, bfloat16 & binary32 & 8 & 360 Tflops/s & Class IV \\
    \bottomrule
  \end{tabular}
  \label{table:mma-hw}
\end{table*}

Table~\ref{table:mma-hw} lists hardware that is available and contains multi-term floating-point addition, as part of dot product and matrix operations.

Most of the companies do not provide information on low level numerical hardware details which makes it hard to classify them and say what numerical features, such as rounding and monotonicity, are present.
An attempt can be made at deducing some of the features from the numerical results that are obtained when computing on these devices, as demonstrated with NVIDIA V100 GPUs by Hickann~and~Bradford~\cite{hibr19} and with V100, T4, and the A100 by Fasi~et~al.~\cite{fhmp21}.
Fasi~et~al.~\cite{fhlm23} have subsequently demonstrated, through error analysis, that low level features such as rounding can become significant when multiplying matrices with matrix arithmetic hardware.

\subsection{Our contributions}

In summary, the present manuscript's contributions to computer arithmetic and beyond are three-fold:
\begin{enumerate}
\item We identify conditions in which floating-point operations that involve multi-term addition can be non-monotonic---this allows to explain surprising numerical results of some of the commercial hardware and construct tests that can be used to look for non-monotonicity of summation within the vector and matrix operations in hardware devices.
  We show that Class IV operations are not monotonic, but Class I-III are and provide proofs in each case.
\item We demonstrate various applications that may be impacted.
\item We propose ways to modify architectures that contain units for adding multiple floating-point numbers in order for the computed approximations of sums to be monotonic.
\item The paper acts as a survey of hardware designs that are and are not monotonic, and fills an important gap in the literature by addressing the monotonicity of multi-term addition.
\end{enumerate}

\section{Background}

\subsection{Floating-point representation and arithmetic}
\label{sec:fp-arith}

We will be using the following IEEE-compliant floating-point
systems and properties.
A binary floating-point number $x$ has the form $(-1)^s \times m \times 2^{e-p+1}$, where $s$ is the sign bit, $p$ is the precision, $m \in [0, \; 2^p-1]$ is the integer significand, and $e \in [e_{\min}, \; e_{\max}]$, with $e_{\min}=1 - e_{\max}$, is the integer exponent.
In order for $x$ to have a unique representation, the number system is \textit{normalized} so that the most significant bit of $m$ is set to~1 if $|x| \ge 2^{e_{\min}}$.
Therefore, all floating-point numbers with $m \geq 2^{p-1}$ are normalized.
Numbers below the smallest normalized number $2^{e_{\min}}$ in absolute value are called \textit{subnormal numbers}, and are such that $e=e_{\min}$ and $0 < m < 2^{p-1}$.
The set of floating-point numbers is denoted by $\mathbb{F}$.

The results of floating-point operations may not be normalized and must be normalized by shifting the significand left or right until it falls within $[2^{p-1}, \; 2^{p}-1]$ and adjusting the exponent accordingly.
Those numbers that cannot be normalized in such a way, due to requiring exponents lower than the minimum exponent value, form subnormal numbers.

The IEEE 754 standard for floating-point arithmetic provides a limited set of requirements for reduction operations such as multi-term addition~\cite[Sec.~9.4]{ieee19} or vector and matrix operations: a particular order of adding the partial sums is not required, and the use of arbitrary precision accumulator is allowed.
The standard does not specify: 1) whether this internal format should be normalized after each addition, 2) which rounding mode should be used, and 3) when the rounding should happen.
IEEE 754-2019~\cite{ieee19} specifies six rounding modes for various purposes.
These requirements provide a lot of freedom in implementation choices and can potentially introduce a wide array of different numerical behaviours.
We identified four main classes of algorithms that are present in literature and made their way into various devices~(Section~\ref{sec:classes}).

\subsection{Monotonicity of IEEE 754 arithmetics}
\label{sec:monotonicity-ieee754}

In this section we demonstrate a few results about the approximation of a sum computed using the 2-term correctly rounded addition operation~\cite{ieee19}.

\subsubsection{Rounding}
\label{sec:rn}

The default rounding mode of IEEE 754 arithmetics is round-to-nearest ties-to-even (RN) and it can be shown that it is monotonic.
Take $x,y \in \mathbb{R}$, $x \leq y$, and two neighbouring floating-point values in some precision-$p$ arithmetic over $\mathbb{F}$, $a$ and $b$.
Assume that $x$ and $y$ lie between $a$ and $b$ such that $a \leq x \leq y \leq b$.
Normalization of the floating-point significand \cite{ieee19} followed by rounding $x \in \mathbb{R}$ to $\mathbb{F}$ is denoted by \fl{x} and in this case $x$ and $y$ can be rounded to $a$ or $b$.

The definition of round-to-nearest does not allow non-monotonic behaviour: since $a \leq x,y \leq b$ and $x \leq y$, $\fl{x} \leq \fl{y}$ because $a = \fl{y} < \fl{x} =b$ contradicts the definition of round-to-nearest \cite{ieee19}.

Other IEEE 754 \cite{ieee19} rounding modes can also be shown to preserve monotonicity: round-toward-zero (RZ), round-toward-negative (RD) and round-toward-positive (RU) are all monotonic because when they are used, $\fl{x}=\fl{y}=a$ or $\fl{x}=\fl{y}=b$.

\subsubsection{Addition of two operands}
\label{sec:addition-of-two-operands}

Now we look at $f_{\mathrm{add}}(x, y) = \fl{x+y}$ where as per IEEE 754, $x+y$ is computed as though in infinite precision arithmetic and then rounded to the nearest value in some $\mathbb{F}$.
Take $a$, $b$, and $c > b$, $s_1 = a + b$, and $s_2 = a + c$, then $s_1 < s_2$.
Using the same approach as in Section~\ref{sec:rn} we can show that $\fl{s_1} \leq \fl{s_2}$ and therefore that the addition of two operands in IEEE 754 arithmetics is also a monotonic function.

\begin{theorem}
  Addition of two operands, $\fl{x+y}$ with $x,y \in \mathbb{R}$, computed using the addition operation as defined in the IEEE 754 is monotonic with round-to-nearest, round-towards-zero, round-toward-negative, and round-toward-positive.
\end{theorem}
\begin{proof}
  Since the addition in IEEE 754 arithmetics is first performed as though in infinite precision, the computed quantities will be $\fl{s_1}$ and $\fl{s_2}$.
  Since $s_1 < s_2$ the reasoning is equivalent to that of Section~\ref{sec:rn} and $\fl{s_1} \leq \fl{s_2}$.
\end{proof}

\subsubsection{Addition of three or more operands}

Multi-term addition in IEEE 754 floating-point arithmetics is also monotonic.
Let \(x_1\), \ldots, \(x_n \in \mathbb{R}\).
Consider $\fl{\fl{x_1+x_2}+x_3}$.
Take $x_3=a$ and set $s_1 = \fl{x_1+x_2}+a$.
Then take $x_3=a+\varepsilon$, with $\varepsilon > \ulp{a}/2$, where $\ulp{a}=2^{e_a-p+1}$ is the size of the gap between $a$ and the following floating-point number, and set $s_2 = \fl{x_1+x_2}+a+\varepsilon$.
Since $a+\varepsilon > a$ we get that $s_1 < s_2$ and we can show that $\fl{\fl{x_1+x_2}+x_3}$ is monotonic by showing that $\fl{s_1} \leq \fl{s_2}$ is, using the reasoning in Sec.~\ref{sec:addition-of-two-operands}.

This can be repeated for showing that $\fl{\fl{\cdots \fl{x_1+x_2} + \cdots} + x_n}$ is also monotonic, and when any of the addends is increased, not necessarily the last one.

\begin{theorem}
  \label{thm:ieee754-sum-monotonic}
  Summation $\sum_i^nx_i$, with $x_i \in \mathbb{R}$ and $n \geq 3$, computed using the floating-point addition operation as defined in the IEEE 754 is monotonic with round-to-nearest, round-towards-zero, round-toward-negative, and round-toward-positive, in any ordering, such as $\fl{\fl{\cdots \fl{x_1+x_2} + \cdots}+x_n}$.
\end{theorem}
\begin{proof}
  First, consider computing $x_1+x_2+\cdots + x_n$ recursively as $\fl{\fl{\cdots \fl{x_1+x_2} + \cdots}+x_n}$.
  Then, consider increasing the last addend $x_n$.
  We can define the partial sum of the first $n-1$ addends as $s_p=\fl{\cdots \fl{x_1+x_2}\cdots)+x_{n-1}}$.
  Then consider two cases, $s_1=\fl{s_p+x_n}$ and $s_2=\fl{s_p+(x_n+\varepsilon)}$.
  Then $s_1 \leq s_2$ through the result in Section~\ref{sec:rn}.

  Secondly, we can check what can happen when any of the addends $x_i$ for $1 \leq i \leq n-1$ are increased before the sum is computed.

  Take $j$ to be the index of an addend which we modify.
  Let
\begin{equation*}
x_i' =
\begin{cases}
x_i + \varepsilon, \quad &i = j,\\
x_i, \quad &i \neq j,
\end{cases}
\end{equation*}
and define the partial sums as
\begin{equation*}
s_i =
\begin{cases}
\fl{x_1},\quad &i = 1,\\
\fl{s_{i-1} + x_i}, \quad &2 \le i \le n,
\end{cases}
\end{equation*}
and
\begin{equation*}
s_i' =
\begin{cases}
\fl{x_1'},\quad &i = 1,\\
\fl{s'_{i-1} + x_i'}, \quad &2 \le i \le n.
\end{cases}
\end{equation*}
We need to prove that \(s_n' \ge s_n\). If \(i < j\), then \(s_i' = s_i\). Using Theorem~\ref{thm:ieee754-sum-monotonic} and the fact that \(x_j' \ge x_j\) we can conclude that
\begin{equation*}
s_j' = \fl{s_{j-1}' + x_j'} = \fl{s_{j-1} + x_j'} \ge \fl{s_{j-1} + x_j} = s_j.
\end{equation*}
For \(j < i \le n\), the result follows by induction:
\begin{equation*}
s_i' = \fl{s_{i-1}' + x_i'} = \fl{s_{i-1}' + x_i} \ge \fl{s_{i-1} + x_i} = s_i.
\end{equation*}
The proof is analogous for other orderings of evaluation of the sum.
\end{proof}

An anonymous referee has pointed out that each ordering can be represented by a tree with vertices representing rounded operations.
For each ordering we need two trees, one with and one without the $\varepsilon$ update to one of the addends $x_j$.
Similarly to the proof above, we can prove monotonicity for each operation and use induction to prove the monotonicity at all levels of the tree as the expression is being computed.
This would allow us to confirm that $s_n' \geq s_n$ which are at the roots of the trees.

\subsubsection{Multiplication}

Now we look at $f_{\mathrm{mul}}(x, y) = \fl{x \times y}$ where as per IEEE 754, $x \times y$ is computed as though in infinite precision arithmetic and then rounded to the nearest value in some $\mathbb{F}$.
Take $a$, $b$, $c > b$, $m_1 = a \times b$, and $m_2 = a \times c$.
If $a>0$, $m_1 < m_2$ (multiplication is monotonic increasing).
If $a < 0$ we have $m_1 > m_2$ (monotonic decreasing).

\begin{theorem}
  \label{thm:mult}
  Multiplication of two operands, $\fl{x\times y}$ with $x \in \mathbb{R}$ and $y \in \mathbb{R}$, computed using the floating-point multiplication operation as defined in the IEEE 754 is monotonic with round-to-nearest, round-towards-zero, round-toward-negative, and round-toward-positive.
\end{theorem}
\begin{proof}
  Since the multiplication in IEEE 754 arithmetics is first performed as though in infinite precision, the computed quantities will be $\fl{m_1}$ and $\fl{m_2}$.
  Since $m_1 < m_2$ the reasoning is equivalent to that of Section~\ref{sec:rn} and $\fl{m_1} \leq \fl{m_2}$.
  The proof is analogous for $a<0$ which gives $\fl{m_1} \geq \fl{m_2}$.
\end{proof}

\begin{theorem}
  The inner product of column vectors $a, b \in \mathbb{R}^n$, $a^Tb$ , computed using the floating-point multiplication and addition operations as defined in IEEE 754 with round-to-nearest, round-towards-zero, round-toward-negative, and round-toward-positive is monotonic for any ordering, such as $\fl{\cdots\fl{\fl{a_1 \times b_1}+\fl{a_2 \times b_2}} + \cdots + \fl{a_n \times b_n}}$.
\end{theorem}
\begin{proof}
  The proof follows from the monotonicity of the scalar multiplication (Theorem~\ref{thm:mult}) and the monotonicity of the $n$-term sum (Theorem~\ref{thm:ieee754-sum-monotonic}).
\end{proof}

Since matrix multiplication is comprised of inner products, elementwise monotonicity results from the monotoncicity of the inner product operation.
This proves the monotonicity properties of the units that lie in Class III (Section~\ref{sec:class3}), which implement multi-term addition hardware to mimic the behaviour of IEEE 754, equivalent to normalizing and rounding after every operation.

\subsection{Fused multi-term adders}

\begin{theorem}
  Summation using fused multi-term adders which perform addition as though in infinite precision and then round once, is monotonic.
\end{theorem}
\begin{proof}
  Since fused multi-term adders compute as though the overall sum is computed in infinite precision and rounded once, $\fl{x_1+x_n+\cdots + x_n}$, they are monotonic due to monotonicity of rounding, showed in Section~\ref{sec:rn}.
\end{proof}
This proves the monotonicity properties of the Class I/II units (Sections~\ref{sec:class1}~and~\ref{sec:class2}).

\section{Results}

In this section we prove a few results about the Class IV multi-term floating-point adders (Section~\ref{sec:class4}).

\subsection{Modified IEEE 754 arithmetics: addition without normalization}
\label{sec:modified-arith}

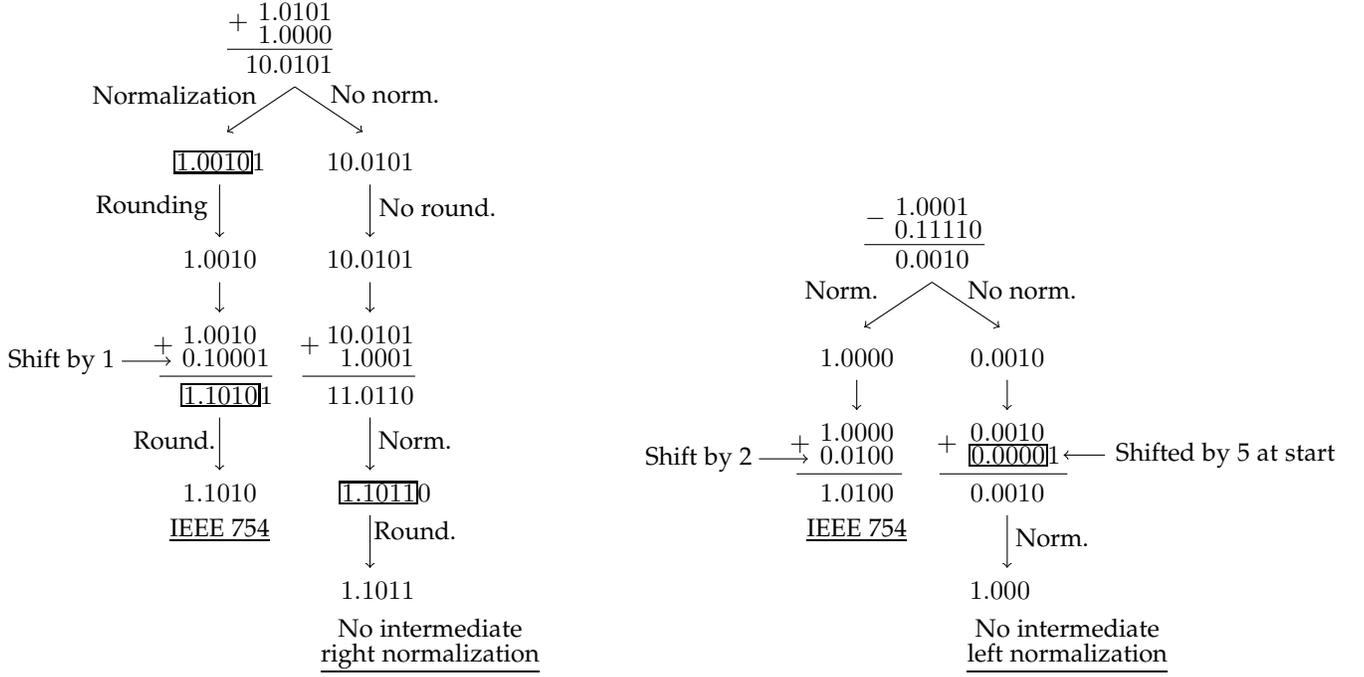
\begin{figure*}
  \centering
  \subfloat{
      \begin{tikzpicture}
        \node (n1) at (0,0) {$1.0101$};
        \node (plus) at (-0.75,-0.15) {$+$};
        \node (n2) at (0,-0.3) {$1.0000$};
        \draw (-0.9,-0.5) -- (0.5,-0.5);
        \node (n3) at (-0.08,-0.7) {$10.0101$};

        \draw [->, black] (0, -1) -- (-0.9, -1.6);
        \node (label) at (-1.6, -1.1) {Normalization};
        \node (n1) at (-1, -2) {$\setlength{\fboxsep}{.2\fboxsep}
\boxed{1.0010}1$};

        \draw [->, black] (-1, -2.3) -- (-1,-3);
        \node (label) at (-1.9, -2.6) {Rounding};
        \node (n1) at (-1, -3.3) {$1.0010$};

        \draw [->, black] (-1, -3.6) -- (-1,-4);
        \node (n1) at (-1,-4.3) {$1.0010$};
        \node (plus) at (-1.75,-4.45) {$+$};
        \node (n2) at (-0.92,-4.6) {$0.10001$};
        \node (label) at (-3.1, -4.65) {Shift by 1};
        \draw [->, black] (-2.3, -4.65) -- (-1.65, -4.65);
        \draw (-1.8,-4.85) -- (-0.3,-4.85);
        \node (n3) at (-0.9,-5.1) {$\setlength{\fboxsep}{.2\fboxsep}\boxed{1.1010}1$};

        \draw [->, black] (-1, -5.4) -- (-1,-6.1);
        \node (label) at (-1.6, -5.7) {Round.};
        \node (n1) at (-1, -6.4) {$1.1010$};
        
        \node (label) at (-1, -6.9) {\underline{IEEE 754}};

        \draw [->, black] (0, -1) -- (0.9, -1.6);
        \node (label) at (1.2, -1.1) {No norm.};
        \node (n1) at (1, -2) {$10.0101$};

        \draw [->, black] (1, -2.3) -- (1,-3);
        \node (label) at (1.9, -2.6) {No round.};
        \node (n1) at (1, -3.3) {$10.0101$};

        \draw [->, black] (1, -3.6) -- (1,-4);
        \node (n1) at (1,-4.3) {$10.0101$};
        \node (plus) at (0.2,-4.45) {$+$};
        \node (n2) at (1.09,-4.6) {$1.0001$};
        \draw (1.6,-4.85) -- (0.1,-4.85);
        \node (n3) at (1,-5.1) {$11.0110$};

        \draw [->, black] (1, -5.4) -- (1,-6.1);
        \node (label) at (1.6, -5.7) {Norm.};
        \node (n1) at (1.2, -6.4) {$\setlength{\fboxsep}{.2\fboxsep}\boxed{1.1011}0$};

        \draw [->, black] (1, -6.7) -- (1,-7.4);
        \node (label) at (1.6, -6.9) {Round.};
        \node (n1) at (1.1, -7.7) {$1.1011$};
        
        \node (label) at (1.8, -8.2) {No intermediate};
        \node (label) at (1.8, -8.6) {\underline{right normalization}};
        
      \end{tikzpicture}}
    \hspace{1cm}
    \subfloat{
      \begin{tikzpicture}
        \node (n1) at (0,0) {$1.0001$};
        \node (plus) at (-0.75,-0.15) {$-$};
        \node (n2) at (0.08,-0.3) {$0.11110$};
        \draw (-0.9,-0.5) -- (0.7,-0.5);
        \node (n3) at (0,-0.7) {$0.0010$};

        \draw [->, black] (0, -1) -- (-0.9, -1.6);
        \node (label) at (-1.2, -1.1) {Norm.};
        \node (n1) at (-1, -2) {$1.0000$};

        \draw [->, black] (-1, -2.3) -- (-1,-2.7);
        \node (n1) at (-1, -3) {$1.0000$};
        \node (plus) at (-1.75,-3.15) {$+$};
        \node (label) at (-3.1, -3.35) {Shift by 2};
        \draw [->, black] (-2.3, -3.35) -- (-1.65, -3.35);
        \node (n2) at (-1,-3.3) {$0.0100$};
        \draw (-1.8,-3.55) -- (-0.4,-3.55);
        \node (n3) at (-1,-3.8) {$1.0100$};

        \node (label) at (-1, -4.3) {\underline{IEEE 754}};

        \draw [->, black] (0, -1) -- (0.9, -1.6);
        \node (label) at (1.2, -1.1) {No norm.};
        \node (n1) at (1, -2) {$0.0010$};

        \draw [->, black] (1, -2.3) -- (1,-2.7);
        \node (n1) at (1, -3) {$0.0010$};
        \node (plus) at (0.2,-3.15) {$+$};
        \node (label) at (3.9, -3.3) {Shifted by 5 at start};
        \draw [->, black] (2.3, -3.3) -- (1.75, -3.3);
        \node (n2) at (1.1,-3.3) {$\setlength{\fboxsep}{.2\fboxsep}\boxed{0.0000}1$};
        \draw (1.8,-3.55) -- (0.1,-3.55);
        \node (n3) at (1,-3.8) {$0.0010$};

        \draw [->, black] (1, -4.1) -- (1,-4.8);
        \node (label) at (1.6, -4.4) {Norm.};
        \node (n3) at (0.9,-5.1) {$1.000$};

        \node (label) at (1.8, -5.6) {No intermediate};
        \node (label) at (1.8, -6) {\underline{left normalization\vphantom{g}}};

      \end{tikzpicture}}
    \caption{Example summation of three precision-$p$ numbers (significands showed) in IEEE 754 arithmetic and a Class IV multi-term adder without the intermediate normalization and rounding. In the multi-term adder the carry bits on the left are kept but the bits past precision $p$ in the fraction are discarded; the normalization and rounding steps are performed at the end, after all the addition operations have been completed. We take $p=5$. On the left is the case in which the significand grows and requires a right-shift to be normalized. On the right is the case with a significant cancellation~\cite[p.~242]{mbdj18} which requires multiple left-shifts to normalize. Notice that the former improves the accuracy of the second addition operation, while the latter makes it worse for the multi-term adder compared with the sum computed using IEEE 754 2-term addition operations. IEEE 754 arithmetic uses round-to-nearest even-on-ties in this example. The monotonicity issue is caused by the lack of right-shift normalization and Equation~\ref{FLR} models the adder that lacks only this normalization in order to simplify.}
    \label{fig:normalization}
  \end{figure*}

We need a modified floating-point addition model to describe Class IV multi-term addition units with precision growth.
We take the normalized significand of a floating-point number to be $2^{p-1} \leq m < 2^{p}$~\cite{ieee19}.
In the binary representation of $m$ the binary point is defined to be between the first and second left-most bits of $m$~\cite{ieee19}.
We now consider a modified version of IEEE 754 addition operation without this constraint, meaning that the normalization step in the addition is not performed.
Specifically, we will focus on the normalization that requires the right-shift of the significand by one step (Figure~\ref{fig:normalization}).
Namely, instead of having one bit to the left of the binary point we assume there are multiple bits for carries to propagate when the result of the partial summation reaches or crosses the powers of two.
Equivalently, we can keep the normalization but add one bit of precision if the sum reaches the next power of two.

Take $a, b \in \mathbb{R}$.
If $|a| < |b|$ swap them so that in general we assure $|a| \geq |b|$.
Define $t =2^{1 + \lfloor \log_2{|a|} \rfloor}$: this finds the absolute value of the power of two nearest to $|a|$ with $|a| < |t|$.
Then the adder with precision increase (which describes an adder without the right-shift normalization) can be defined as
\begin{align}\label{FLR}
    \flr{a+b} &=
    \begin{cases}
      \flp{p}{a+b} & \text{ if } |a+b| < t, \\
      \flp{p+1}{a+b} & \text{ if } |a+b| \geq t.
    \end{cases}
\end{align}
When we use this adder multiple times, for example to compute $\flr{\flr{x_1 + x_2}+x_3}$ any precision increase in the first adder is propagated into the next adder.
Therefore this expression can grow precision from $p$ to $p+1$, while $n$ additions could grow precision from $p$ to $p+\lceil\log_2n\rceil$---we call this precision growth after each addition in a multi-term summation a \textit{gradual precision growth}.
In practice the final result may also be rounded to some desired target precision: $\fl{\flr{x_1+x_2}}$, $\fl{\flr{\flr{x_1+x_2}+x_3}}$, and so on.
As an aside, this double rounding can cause issues with accuracy of the final result~\cite{roux14, rump17}, but this is not the cause of the non-monotonicity and is not addressed further in this paper.
This model of addition does not model the lack of left-shift normalization (Figure~\ref{fig:normalization}); it is not required for the purposes of this article.

A similar device was used by Ashenhurst~and~Metropolis~\cite[p.~418]{asme59} for error analysis of unnormalized floating-point arithmetic.

\subsection{Monotonicity of the modified addition}

It can be shown using the similar reasoning as in Section~\ref{sec:monotonicity-ieee754} that $\fl{\flr{x + y}}$ and $\fl{\flr{\flr{x_1+x_2}+x_3}}$ are monotonic.

\begin{theorem}
  Addition of two operands, $\flr{x+y}$ with $x \in \mathbb{R}$ and $y \in \mathbb{R}$, is monotonic with round-to-nearest, round-towards-zero, round-toward-negative, and round-toward-positive.
\end{theorem}
\begin{proof}
  First, if the internal adder does not grow precision, the final rounding does not have any effect and the summations are monotonic as shown in Section~\ref{sec:monotonicity-ieee754}.
  Let us consider the monotonicity of $\fl{\flr{x+y}}$ when the precision grows by one bit.
  Take $s_1=\flr{a+b}$ and $s_2=\flr{a+c}$ with $c>b$.
  Due to monotonicity of rounding, $s_1\leq s_2$, and therefore $\fl{s_1}\leq\fl{s_2}$.
\end{proof}

\begin{theorem}
  Addition of three operands, $\flr{\flr{x_1+x_2}+x_3}$ with $x_i \in \mathbb{R}$, is non-monotonic with round-to-nearest, round-towards-zero, and round-toward-negative ($x_i > 0$) or round-toward-positive ($x_i < 0$), except if the final rounding $\fl{\flr{\flr{x_1+x_2}+x_3}}$ to the starting precision is performed.
  \label{thm:three-term}
\end{theorem}
\begin{proof}
  Take $a$, $b$, and $c$ where $b$ is a power of two, $a$ is the floating-point number preceding $b$, and $c$ is the floating-point number following $b$, such that $a<b<c$.
  We consider positive values, but the proof for negative values is analogous.
  Also, take $\varepsilon=\frac{c-b}{2}$.
  Then, $\flr{b+\varepsilon}=b$ with RN, RD, and RZ, as is $\flr{\flr{b+\varepsilon}+\varepsilon}=b$.
  However, $\flr{a+\varepsilon}=b$ and precision grows by one bit.
  Due to precision growth, $\flr{\flr{a+\varepsilon}+\varepsilon} > b$.
  Therefore monotonicity is not preserved.
  However, the final rounding $\fl{\flr{\flr{a+\varepsilon}+\varepsilon}}=b$ and overall the monotonicity is preserved.
  With RU monotonicity is present even without the final rounding.
\end{proof}

However, as we now show, a sum that includes $n > 3$ terms computed with non-normalized additions modelled by Equation~\ref{FLR} can be non-monotonic in general.

\begin{theorem}
  Summation $\flr{\cdots \flr{x_1+x_2} + \cdots)+x_n}$, with $x_i \in \mathbb{R}$ and $n \geq 4$ is not monotonic with round-to-nearest, round-towards-zero, and round-toward-negative ($x_i > 0$) or round-toward-positive ($x_i < 0$), with and without the final rounding to the starting precision.
  \label{thm:sum-nonmonotonic}
\end{theorem}
\begin{proof}
Take three consecutive positive floating-point values in some precision-$p$ arithmetic, $a$, $b$, and $c$ with $b$ a power of two.
Then consider evaluating a 4-term summation $\flr{\flr{\flr{x + \varepsilon} + \varepsilon} + \varepsilon}$ with $x, \varepsilon > 0$ (similar example can be shown for $x, \varepsilon < 0$).
In precision-$p$ arithmetic, with round to nearest ties to even, we can show that $\flr{b+\varepsilon}=b$ for $\varepsilon \leq (c-b)/2$, while in precision-$(p+1)$ arithmetic $\flr{b+\varepsilon}=b$ for $\varepsilon \leq (c-b)/4$.
Also, in precision-$p$ arithmetic $a+(c-b)/2=b$.

Take $\varepsilon = (c-b)/2$ and consider two cases.
\begin{enumerate}
\item $x = b$, then $\flr{\flr{\flr{b + \varepsilon} + \varepsilon} + \varepsilon} = b$ (all in precision-$p$).
\item $x = a$, then the first addition $\flr{a + \varepsilon} = b$ (and precision increases to $p+1$ since $b$ is a power of two). Following that, the second addition $\flr{b+\varepsilon}= b + \varepsilon$ as well as the third addition $\flr{b + \varepsilon + \varepsilon} = c$ (since we are in precision-($p+1$)).
\end{enumerate}
Since the sum evaluates to $b$ when $x=b$ and to $c$ when $x = a < b$, we have shown that the 4-term sum in this modified arithmetic is non-monotonic.
The final rounding would not change the result because  $\fl{\flr{b + \varepsilon + \varepsilon}} = c$ since $2\varepsilon$ is a value stored in the bits to the left of the rounding point.
\end{proof}

\begin{corollary}
    The inner product of vectors $a, b \in \mathbb{R}^n$, $\fl{\cdots\flr{\fl{a_1 \times b_1}+\fl{a_2 \times b_2}} + \cdots + \fl{a_n \times b_n}}$ for $n \geq 4$, with round-to-nearest, round-toward-zero, and round-toward-negative ($a_i \times b_i > 0$) or round-toward-positive ($a_i \times b_i < 0$) is non-monotonic.
\end{corollary}
\begin{proof}
  Set all elements of $b$ to $1$. Then the proof of Theorem~\ref{thm:sum-nonmonotonic} concludes this proof.
\end{proof}
\begin{corollary}
  Matrix-vector multiplication $Ax$ and matrix-matrix multiplication $AB$, where $A\in \mathbb{R}^{m\times n}$, $x \in \mathbb{R}^{n}$, $B \in \mathbb{R}^{n \times l}$, $n \geq 4$, with round-to-nearest, round-toward-zero, and round-toward-negative ($a_{ik} \times b_{kj} > 0$) or round-toward-positive ($a_{ik} \times b_{kj} < 0$) are element-wise non-monotonic.
\end{corollary}
\begin{proof}
  Each element of the output vector or matrix is computed by the inner product.
\end{proof}

\subsection{Impact of the order of addition}

\begin{figure*}[th!]
  \begin{center}
    \begin{tikzpicture}
      \begin{groupplot}[
        group style={
          group size=4 by 1,
        },
        width=1.8in,
        grid=major,
        ]
        
        \nextgroupplot[
        ylabel={$\frac{|s-\widehat{s}|}{s}$},
        title={4-term adder},
        xlabel = {$n$}
        ]
        
        \addplot[color=OliveGreen!70, mark=triangle] table [x=length, y=fp16-inc-ord] {experiments/data/compare_summation_algs_terms4.dat};
        \addplot[color=black, mark=square] table [x=length, y=fp16-dec-ord] {experiments/data/compare_summation_algs_terms4.dat};
        \addplot[color=Fuchsia!70, mark=o] table [x=length, y=fp16-multi-term-add] {experiments/data/compare_summation_algs_terms4.dat};

        \nextgroupplot[
        title={64-term adder}
        ]

        \addplot[color=OliveGreen!70, mark=triangle] table [x=length, y=fp16-inc-ord] {experiments/data/compare_summation_algs_terms64.dat};
        \addplot[color=black, mark=square] table [x=length, y=fp16-dec-ord] {experiments/data/compare_summation_algs_terms64.dat};
        \addplot[color=Fuchsia!70, mark=o] table [x=length, y=fp16-multi-term-add] {experiments/data/compare_summation_algs_terms64.dat};

        \nextgroupplot[
        title={512-term adder}
        ]

        \addplot[color=OliveGreen!70, mark=triangle] table [x=length, y=fp16-inc-ord] {experiments/data/compare_summation_algs_terms512.dat};
        \addplot[color=black, mark=square] table [x=length, y=fp16-dec-ord] {experiments/data/compare_summation_algs_terms512.dat};
        \addplot[color=Fuchsia!70, mark=o] table [x=length, y=fp16-multi-term-add] {experiments/data/compare_summation_algs_terms512.dat};

        \nextgroupplot[
        title={1024-term adder}
        ]

        \addplot[color=OliveGreen!70, mark=triangle] table [x=length, y=fp16-inc-ord] {experiments/data/compare_summation_algs_terms1024.dat};
        \addplot[color=black, mark=square] table [x=length, y=fp16-dec-ord] {experiments/data/compare_summation_algs_terms1024.dat};
        \addplot[color=Fuchsia!70, mark=o] table [x=length, y=fp16-multi-term-add] {experiments/data/compare_summation_algs_terms1024.dat};
        
      \end{groupplot}
    \end{tikzpicture}
    
    \begin{tikzpicture}[trim axis left, trim axis right]
      \begin{axis}[
        title = {},
        legend columns=3,
        scale only axis,
        width=1mm,
        hide axis,
        /tikz/every even column/.append style={column sep=0.4cm},
        legend style={at={(0,0)},anchor=center,draw=none,
          legend cell align={left},cells={line width=0.75pt}},
        legend image post style={sharp plot},
        legend cell align={left},
        ]
        \addplot [mark=triangle, OliveGreen!70] (0,0);
        \addplot [mark=square, black] (0,0);
        \addplot [mark=o, Fuchsia!70] (0,0);
        \legend{Inc. order, Dec. order, Use of multi-term adder};
      \end{axis}
    \end{tikzpicture}
  \end{center}
  \caption{Summation of positive random binary16 vectors of increasing length. Three summation algorithms are used: recursive summation in increasing order of magnitude using the binary16 IEEE 754 arithmetic, recursive summation in decreasing order of magnitude using the binary16 IEEE 754 arithmetic, and recursive blocking summation using Class IV multi-term adders of various sizes. Relative errors are measured by comparing with the summation of the same values in binary64 arithmetic.}
  \label{fig:summation_algorithms0}
\end{figure*}
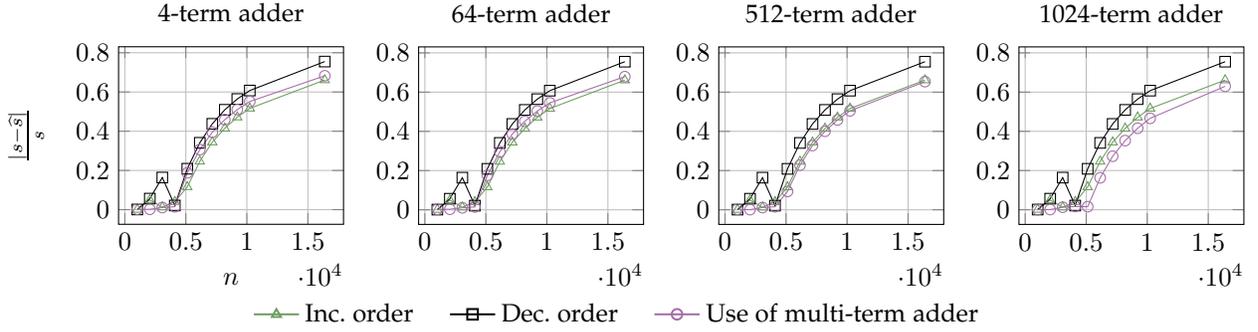

Class IV multi-term adders align all input significands relative to the largest magnitude addend.
This is performed so that all the significands can be right-shifted, with lower order bits dropped or rounded.
If the alignments were performed relative to an addend of arbitrary choice, shifting both to the left and right would be required.
When the shifts are to the left, shifted bits cannot be dropped and would have to be preserved, which would introduce a high hardware cost similar to Class I/II adders.
As many bits from the left shifts would have to be preserved as the highest difference between the exponents of addends.

After the significands are aligned relative to the largest magnitude addend's exponent, if there is no normalization of intermediate sums, the addition acts like the addition of fixed-point values and is associative; the order of performing additions in such a method will not impact the final result.
Whether all the addends are added in series or in parallel also does not have any impact on the final result.

When adding a series of positive floating-point values, doing so in an increasing rather than a decreasing order in absolute value reduces the worst-case error bound~\cite[Sec.~4.2]{high:ASNA2}; this ordering may not reduce the \emph{actual error}~\cite[Sec.~2]{high93s}.
When considering both positive and negative numbers, the decreasing order can yield better accuracy in the face of cancellation~\cite[Sec.~2]{high93s}.
However, the hardware starts at the largest magnitude addend, as discussed above; therefore it can be interesting to check the accuracy of the summation, which we do in the following section.
Note that that when there is no intermediate normalization implemented, the addition is in effect growing precision, which may improve the accuracy.

\section{Numerical experiments}

We have simulated the Class IV multi-term adders with the gradual precision growth in MATLAB, using the custom precision simulator \texttt{CPFloat}~\cite{fami23}.
The code for reproducing the results is available\footnote{\url{https://github.com/north-numerical-computing/multi-operand-add-monotonicity}}.

\subsection{Order of addends and associativity}

In Figure~\ref{fig:summation_algorithms0} we plot relative errors $\frac{|s_n-\widehat{s_n}|}{s_n}$ for adding $n$ floating-point values in binary16 arithmetic ($\widehat{s_n}$), with $p=11$ and $e_{max}=15$, compared with the sum in binary64 arithmetic ($s_n$).
The addends are pseudo-random numbers generated in MATLAB using the \emph{mrg32k3a} random number generator with a seed of $500$, in the range $(0,0.001)$.
We vary $n$ and perform summations with recursive summation algorithm with addends in the decreasing and increasing orders, as well as in the original order with the multi-term adders with terms $4$, $64$, $512$, and $1024$.
As expected, the decreasing ordering results in the largest errors.
In most cases, multi-term adders are worse or very close to recursive summation with increasing ordering of addends.
However, with the 1024-term adder, sums become more accurate.
This can be explained through precision growth---with the larger adders there is a possibility for more precision growth, which improves accuracy.
In this experiment we observed precision to grow to the total of $19$ bits, from the default $11$.

In \texttt{test\_associativity.m} we generated a vector of 64 pseudo-random binary16 values, randomly permuted them $10^4$ times and each time computed the sum in IEEE 754 arithmetic and the model of a 64-term Class IV adder.
Checking the range of computed sums we found that it is non-zero for the IEEE 754 arithmetic and zero for the 64-term Class IV adder, confirming non-associativity and associativity, respectively.

\subsection{Monotonicity}

For the purposes of demonstrating non-monotonicity of the Class IV adders, we compute $$\fl{\cdots \fl{x_1+x_2} + \cdots)+x_n}$$ and $$\fl{\flr{\cdots \flr{x_1+x_2} + \cdots)+x_n}}$$ in three small floating-point systems: $p=3$, $e_{max}=3$; $p=4$, $e_{max}=3$; and $p=5$, $e_{max}=4$.
We construct the sum for severe non-monotonicity to appear, as follows.
First, we set all $x_i=0.25$ and then vary $x_1$ by changing it to the adjacent floating-point value towards $+\infty$ until all representable values are covered.
On each iteration we sum the values $x_i$ with the two different addition models, the multi-term adder with precision growth (Class IV) and the IEEE 754 adder with normalization and rounding after each addition operation (Class III).
We report the value of the sum as well as the relative error compared with the same sum performed in binary64 arithmetic.
The results are plotted in Figure~\ref{fig:monotonicity-experiment0}.

First, consider the first column of diagrams in Figure~\ref{fig:monotonicity-experiment0}.
In the top diagram, we see that from the beginning the sum saturates to some quantity and for a while stagnates with the IEEE 754 arithmetic.
Relative error starts increasing.
When the sum reaches this point, all remaining addends (set to $0.25$) are rounded down and do not contribute to the sum.
With the multi-term adder this does not occur because precision grows on powers of two, allowing the sum to keep changing as $x_1$ is being increased.
At a certain point, when $x_1$ crosses the value at which the sum stagnates, the overall value of the sum becomes $x_1$ and both arithmetics align.
At the beginning of this, non-monotonicity appears in the multi-term addition.
As precision is increased (rows of the matrix of diagrams in Figure~\ref{fig:monotonicity-experiment0}), the point at which the IEEE 754 stagnates, and the point at which the multi-term addition shows non-monotonicity, moves to higher values of the sum.

Other columns in Figure~\ref{fig:monotonicity-experiment0} correspond to the larger number of terms being added.
The main observation is that with more terms the severity of non-monotonicity increases because a larger number of terms can grow the sum more in the range where IEEE 754 arithmetic stagnates and the multi-term adder grows precision.

\begin{figure*}[ph!]
  \begin{center}    
    \begin{tikzpicture}
      \node[draw, rotate=90] at (-2,1) {$p=3$ \quad $e_{max}=3$};
      \node[draw, rotate=90] at (-2,-5) {$p=4$ \quad $e_{max}=3$};
      \node[draw, rotate=90] at (-2,-11.5) {$p=5$ \quad $e_{max}=4$};
      \begin{groupplot}[
        group style={
          group size=3 by 6,
        },
        width=2.2in,
        grid=major,
        title style={yshift=-0.1in},
        ]

        \nextgroupplot[
        ylabel = {Sum},
        title={8 terms},
        xmajorticks=false,
        ymin=0
        ]
        
        \addplot[color=Fuchsia!70, mark=o] table [x=x1, y=sum] {experiments/data/summation_values_p3emax3T8.dat};
        \addplot[color=black, mark=diamond, only marks] table [x=x1, y=sum-ieee754] {experiments/data/summation_values_p3emax3T8.dat};

        \nextgroupplot[
        title={16 terms},
        xmajorticks=false,
        ymin=0
        ]

        \addplot[color=Fuchsia!70, mark=o] table [x=x1, y=sum] {experiments/data/summation_values_p3emax3T16.dat};
        \addplot[color=black, mark=diamond, only marks] table [x=x1, y=sum-ieee754] {experiments/data/summation_values_p3emax3T16.dat};

        \nextgroupplot[
        title={32 terms},
        xmajorticks=false,
        ymin=0
        ]
        
        \addplot[color=Fuchsia!70, mark=o] table [x=x1, y=sum] {experiments/data/summation_values_p3emax3T32.dat};
        \addplot[color=black, mark=diamond, only marks] table [x=x1, y=sum-ieee754] {experiments/data/summation_values_p3emax3T32.dat};

        \nextgroupplot[
        height=1.2in,
        yshift=0.8cm,
        ylabel={Rel. error},
        xlabel={$x_1$},
        ]

        \addplot[color=Fuchsia!70, mark=o] table [x=x1, y=sum-error] {experiments/data/summation_values_p3emax3T8.dat};
        \addplot[color=black, mark=diamond, only marks] table [x=x1, y=sum-ieee754-error] {experiments/data/summation_values_p3emax3T8.dat};

        \nextgroupplot[
        height=1.2in,
        yshift=0.8cm,
        xlabel={$x_1$},
        ]
        
        \addplot[color=Fuchsia!70, mark=o] table [x=x1, y=sum-error] {experiments/data/summation_values_p3emax3T16.dat};
        \addplot[color=black, mark=diamond, only marks] table [x=x1, y=sum-ieee754-error] {experiments/data/summation_values_p3emax3T16.dat};

        \nextgroupplot[
        height=1.2in,
        yshift=0.8cm,
        xlabel={$x_1$},
        ]
        \addplot[color=Fuchsia!70, mark=o] table [x=x1, y=sum-error] {experiments/data/summation_values_p3emax3T32.dat};
        \addplot[color=black, mark=diamond, only marks] table [x=x1, y=sum-ieee754-error] {experiments/data/summation_values_p3emax3T32.dat};

        \nextgroupplot[
        yshift=-0.3cm,
        xmajorticks=false,
        title={8 terms},
        ylabel = {Sum},
        ymin=0
        ]

        \addplot[color=Fuchsia!70, mark=o] table [x=x1, y=sum] {experiments/data/summation_values_p4emax3T8.dat};
        \addplot[color=black, mark=diamond, only marks] table [x=x1, y=sum-ieee754] {experiments/data/summation_values_p4emax3T8.dat};

        \nextgroupplot[
        xmajorticks=false,
        title={16 terms},
        yshift=-0.3cm,
        ymin=0
        ]

        \addplot[color=Fuchsia!70, mark=o] table [x=x1, y=sum] {experiments/data/summation_values_p4emax3T16.dat};
        \addplot[color=black, mark=diamond, only marks] table [x=x1, y=sum-ieee754] {experiments/data/summation_values_p4emax3T16.dat};

        \nextgroupplot[
        xmajorticks=false,
        title={32 terms},
        yshift=-0.3cm,
        ymin=0
        ]

        \addplot[color=Fuchsia!70, mark=o] table [x=x1, y=sum] {experiments/data/summation_values_p4emax3T32.dat};
        \addplot[color=black, mark=diamond, only marks] table [x=x1, y=sum-ieee754] {experiments/data/summation_values_p4emax3T32.dat};

        \nextgroupplot[
        height=1.2in,
        yshift=0.8cm,
        ylabel={Rel. error},
        xlabel={$x_1$},
        ]
        
        \addplot[color=Fuchsia!70, mark=o] table [x=x1, y=sum-error] {experiments/data/summation_values_p4emax3T8.dat};
        \addplot[color=black, mark=diamond, only marks] table [x=x1, y=sum-ieee754-error] {experiments/data/summation_values_p4emax3T8.dat};

        \nextgroupplot[
        height=1.2in,
        yshift=0.8cm,
        xlabel={$x_1$},
        ]
        
        \addplot[color=Fuchsia!70, mark=o] table [x=x1, y=sum-error] {experiments/data/summation_values_p4emax3T16.dat};
        \addplot[color=black, mark=diamond, only marks] table [x=x1, y=sum-ieee754-error] {experiments/data/summation_values_p4emax3T16.dat};

        \nextgroupplot[
        height=1.2in,
        yshift=0.8cm,
        xlabel={$x_1$},
        ]
        
        \addplot[color=Fuchsia!70, mark=o] table [x=x1, y=sum-error] {experiments/data/summation_values_p4emax3T32.dat};
        \addplot[color=black, mark=diamond, only marks] table [x=x1, y=sum-ieee754-error] {experiments/data/summation_values_p4emax3T32.dat};

        \nextgroupplot[
        xmajorticks=false,
        title={8 terms},
        ylabel = {Sum},
        yshift=-0.3cm,
        ymin=0
        ]

        \addplot[color=Fuchsia!70, mark=o] table [x=x1, y=sum] {experiments/data/summation_values_p5emax4T8.dat};
        \addplot[color=black, mark=diamond, only marks] table [x=x1, y=sum-ieee754] {experiments/data/summation_values_p5emax4T8.dat};

        \nextgroupplot[
        xmajorticks=false,
        title={32 terms},
        yshift=-0.3cm,
        ymin=0
        ]

        \addplot[color=Fuchsia!70, mark=o] table [x=x1, y=sum] {experiments/data/summation_values_p5emax4T32.dat};
        \addplot[color=black, mark=diamond, only marks] table [x=x1, y=sum-ieee754] {experiments/data/summation_values_p5emax4T32.dat};

        \nextgroupplot[
        xmajorticks=false,
        title={64 terms},
        yshift=-0.3cm,
        ymin=0
        ]
        
        \addplot[color=Fuchsia!70, mark=o] table [x=x1, y=sum] {experiments/data/summation_values_p5emax4T64.dat};
        \addplot[color=black, mark=diamond, only marks] table [x=x1, y=sum-ieee754] {experiments/data/summation_values_p5emax4T64.dat};

        \nextgroupplot[
        xlabel={$x_1$},
        height=1.2in,
        yshift=0.8cm,
        ymax=0.23,
        ylabel={Rel. error},
        ]

        \addplot[color=Fuchsia!70, mark=o] table [x=x1, y=sum-error] {experiments/data/summation_values_p5emax4T8.dat};
        \addplot[color=black, mark=diamond, only marks] table [x=x1, y=sum-ieee754-error] {experiments/data/summation_values_p5emax4T8.dat};

        \nextgroupplot[
        xlabel={$x_1$},
        height=1.2in,
        yshift=0.8cm
        ]

        \addplot[color=Fuchsia!70, mark=o] table [x=x1, y=sum-error] {experiments/data/summation_values_p5emax4T32.dat};
        \addplot[color=black, mark=diamond, only marks] table [x=x1, y=sum-ieee754-error] {experiments/data/summation_values_p5emax4T32.dat};

        \nextgroupplot[
        xlabel={$x_1$},
        height=1.2in,
        yshift=0.8cm
        ]

        \addplot[color=Fuchsia!70, mark=o] table [x=x1, y=sum-error] {experiments/data/summation_values_p5emax4T64.dat};
        \addplot[color=black, mark=diamond, only marks] table [x=x1, y=sum-ieee754-error] {experiments/data/summation_values_p5emax4T64.dat};
      \end{groupplot}
    \end{tikzpicture}
    
    \begin{tikzpicture}[trim axis left, trim axis right]
      \begin{axis}[
        title = {},
        legend columns=2,
        scale only axis,
        width=1mm,
        hide axis,
        /tikz/every even column/.append style={column sep=0.4cm},
        legend style={at={(0,0)},anchor=center,draw=none,
          legend cell align={left},cells={line width=0.75pt}},
        legend image post style={sharp plot},
        legend cell align={left},
        ]
        \addplot [mark=o, Fuchsia!70] (0,0);
        \addplot [mark=diamond, black] (0,0);
        \legend{Multi-term adder, IEEE 754 adder};
      \end{axis}
    \end{tikzpicture}
  \end{center}
  \caption{Summation with various floating-point formats and two types of addition: multi-term addition with precision growth (Class IV) and IEEE 754 addition with normalization and rounding after each operation. Value of of the sum (top of each subdiagram) and the relative error (bottom) compared with the sum performed in the binary64 arithmetic. The $x$-axis corresponds to the quantity of $x_1$ which we vary, with the rest of $x_i=0.25$.}
  \label{fig:monotonicity-experiment0}
\end{figure*}
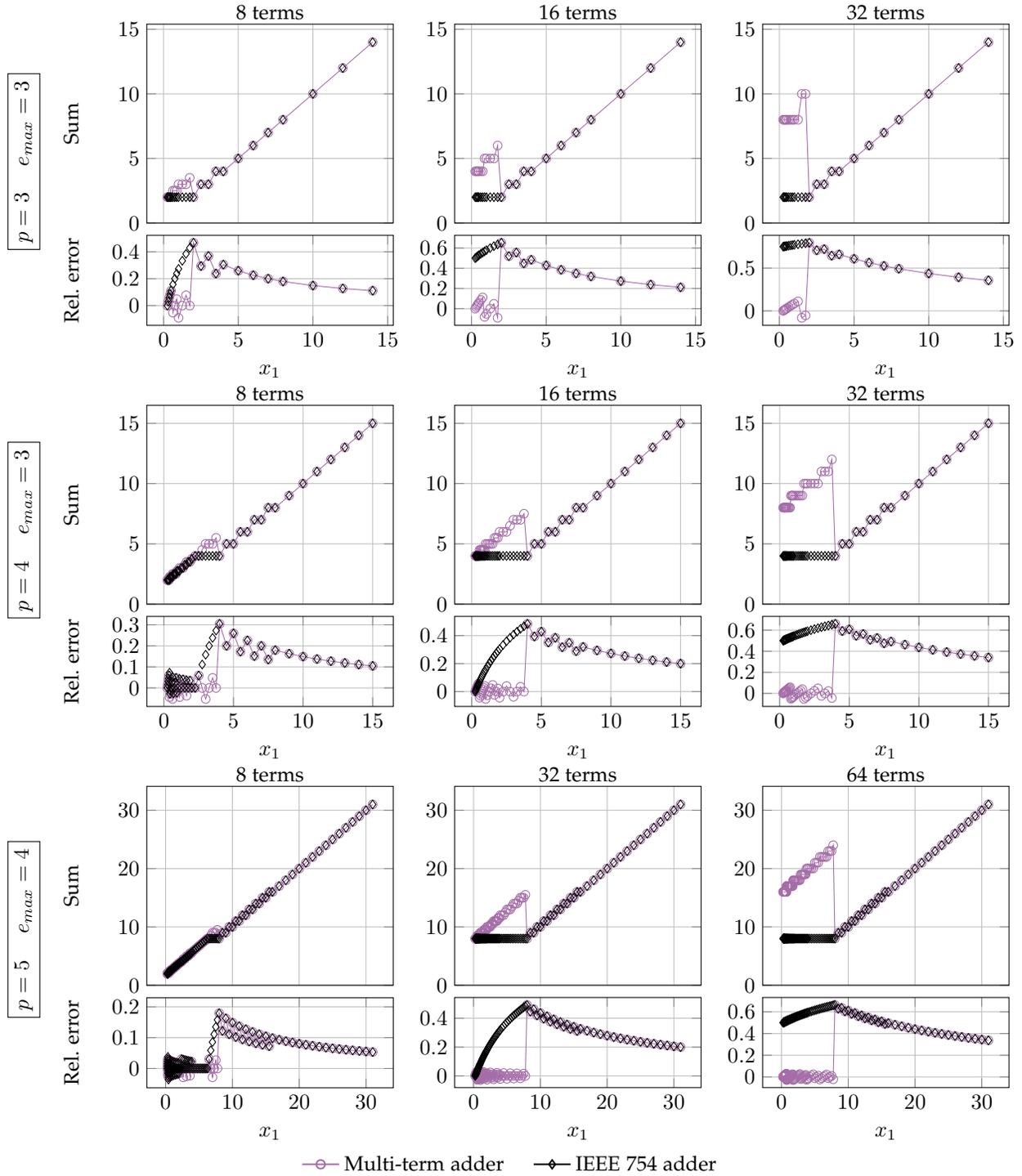

\section{Discussion}

\subsection{Where non-monotonicity can cause issues}

In this subsection we discuss various algorithms that may be impacted by the non-monotonicity of floating-point arithmetic.

In a 1967 paper, Forsythe~\cite{fors67} notes that despite the lack of associativity and distributivity in floating-point addition and multiplication, one can do good analysis provided only that the arithmetic is monotonic.
At that time monotonicity was not always preserved in arithmetics, like it was after the IEEE 754 was released: ``\emph{Such properties seem elemental but they are extremely helpful. And they are surprisingly absent!''~\cite[Sec.~4]{fors67}}.

Demmel,~Dhillon,~and Ren~\cite{ddr95} explore the bisection algorithm for finding eigenvalues of real symmetric tridiagonal matrices.
The algorithm relies on the $\mathrm{count}(x)$ to count the number of eigenvalues that are smaller than $x$ (Algorithm~\ref{alg:count-eigv}).
This function is monotonic, but when implemented in floating-point arithmetic, various implementation details can introduce non-monotonicity, which in turn can cause wrong results by finding negative number of eigenvalues $(\mathrm{count}(x_b)-\mathrm{count}{(x_a)}) < 0$ over a specific range $x_a$ to $x_b$ with $x_b>x_a$.
We have shown that non-monotonicity can appear in the three-term adder if no final rounding is performed (Theorem~\ref{thm:three-term}), which would invalidate the Assumption~1A~\cite[p~120]{ddr95} and introduce a source of non-monotonicity in an implementation of $\mathrm{count}(x)$.

\begin{algorithm2e}[t]
  \caption{$\mathrm{Count}(x, T)$: return the number of eigenvalues of a real symmetric tridiagonal matrix $T$.}
  \label{alg:count-eigv}
  $count \gets 0$\;
  $d \gets 1$\;
  \For{$i \gets 1$ \KwTo $n$}{
    $d \gets a_i - x - \frac{b_{i-1}^2}{d}$\;
    \If{$d < 0$}{
      $count \gets count + 1$\;
    }
  }
  \Return{$count$}\;
\end{algorithm2e}

Higham~\cite[Sec.~2.6]{high:ASNA2} demonstrates that monotonicity of rounding may be an important property when solving a quadratic equation $ax^2-2bx+c=0$.
Computing an expression $\sqrt{b^2-ac}$ is required and if $b^2 = ac$ but $\fl{b^2} < b^2$, computing $\fl{\fl{b^2}-ac}$ with an FMA can result in a negative number passed into the square root function.
While this would not require a 4-term or wider adder, which we explore in this paper, if an expression $\sqrt{\sum_{i=1}^{n}a_i-\sum_{i=1}^{n}b_i}$, where $a$ and $b$ are vectors, which uses the multi-term adder twice and then subtracts the sums, appears in applications, a similar issue to that discussed by Higham can appear.
We provide an example in \texttt{test\_sqrt.m}.
We use a binary32 arithmetic with $p=24$ and $e_{max}=127$, and $n=8$.
Then take
$$a=[1, 1, 1, 1, 1, 1, 1, 16777216]$$
and
$$b=[1, 1, 1, 1, 1, 1, 1, 16777214].$$
Computing sums recursively left to right with IEEE 754 additions in order we get $\sqrt{4}$, while computing them with the 8-term Class IV adder we get $\sqrt{-4}$ due to the non-monotonicity in the sum.
The correct result, computed in binary64 arithmetic, is $\sqrt{2}$.
Note that with the IEEE 754 binary32 arithmetic we can change the order of addends and impact the result:
$$a=[16777216, 1, 1, 1, 1, 1, 1, 1]$$
$$b=[16777214, 1, 1, 1, 1, 1, 1, 1]$$
gives $\sqrt{0}$, but
$$a=[16777216, 1, 1, 1, 1, 1, 1, 1]$$
$$b=[1, 1, 1, 1, 1, 1, 1, 16777214]$$
gives $\sqrt{-4}$.
However, as discussed, with the Class IV adder the order does not impact the final result and this example results in $\sqrt{-4}$ irrespective of it; as a result it cannot be fixed at software layer by reordering.

In computing elementary functions often monotonicity is one of the properties that is desired~\cite{mull16}.
Silverstein~et~al.~\cite{sst90} in the report on the UNIX math library \texttt{libm}, note: ``\emph{Monotonicity failures can cause problems, for example, in evaluating divided diferences. The only function we have observed to violate monotonicity for C/SVR4 is \texttt{lgamma}. We would not be surprised to learn of others.}''.
This is of general interest, and we are not aware of a link between the non-monotonicity of multi-term adders and mathematical functions.

Interval arithmetic can be impacted by non-monotonicity of floating-point.
Rump~\cite[Sec.~2.1]{rump99} gives an example of computing lower and upper limits of the interval of matrix multiplication of point matrices (lower and upper limits are equal) through the change of rounding modes in floating-point arithmetic.
Similar to Rump's example, in \texttt{test\_interval.m} we compute the interval of summation operation, through round-toward-negative and round-toward-positive, for lower and upper limits, respectively.
First we sum a vector $a=[16777216, 1, 1, 1, 1, 1, 1, 1]$ which yields an interval $[16777216,16777230]$.
We then decrease $a_1$ and perform the same for $b=[16777214, 1, 1, 1, 1, 1, 1, 1]$ which yields an interval $[16777220,16777222]$.
By decreasing one of the arguments we expect the interval to shift down the real axis, but the lower end shifts up and the interval becomes narrower due to precision growth in the multi-term adder.
As before, it is possible to cause a similar issue with the IEEE 754 arithmetic, by changing the order of the addends.

\subsection{Possible solutions}

Having gradual precision growth in the dot product can be beneficial in getting more accurate results.
However, if monotonicity is needed, a straightforward solution would be to always normalize after each addition to stop the gradual precision growth in the internal accumulator.
This can replicate the behaviour of a software implementation of summation based on standard IEEE 754 operations.
The hardware cost would most likely be increased substantially.

Another approach could be to detect where precision growth occurs by monitoring the carry-out bits in the significands of the sums and then informing any further additions to cancel an appropriate number of bits from the bottom of the significand and not take them into account when adding.
However, this is highly dependent on what particular implementations are doing and would probably impact \textit{guard}, \textit{round} and \textit{sticky} bits which might or might not be used in the intermediate additions inside the multi-term adder, depending on the rounding properties needed.
The additional logic for this may make the Class IV adders as expensive as Class I/II adders, but we leave this for a separate hardware-oriented study.

As a summary:
\begin{itemize}
\item Class I/II adders perform the summation as though only one rounding error is induced at the end, and with these adders monotonicity and associativity of summation are preserved.
\item Class III adders perform the summation which numerically behaves as a software implementation with the IEEE 754 addition operations, with normalization and rounding after each.
  With these adders, associativity is not preserved, but monotonicity is.
\item With the Class IV adders, monotonicity is not preserverd, as we have demonstrated, but the associativity is.
\end{itemize}
One can choose an appropriate implementation based on the numerical properties required in a particular architecture.

\section{Conclusion}

IEEE 754-2019 provides the following guidance for implementing multi-term summation or dot product operations (collectively called \emph{reduction operations}) \cite[Sec.~9.4]{ieee19}: ``\emph{Language standards should define the following reduction operations for all supported arithmetic formats. Unlike the other operations in this standard, these operate on vectors of operands in one format and return a result in the same format. Implementations may associate in any order or evaluate in any wider format.}''.
In the present manuscript we analysed how various hardware designs in the literature as well as in the available hardware implement this.
We demonstrated that there are four classes of implementation, each with different hardware complexity and numerical properties.

Focusing on Class IV multi-term addition, where the significand alignment of the summands is performed in limited precision and the sum is performed without the normalization, all resulting in precision growth during the summation, we proved that non-monotonicity can occur.
We also showed that Class I--III devices are not subject to this.
Our results should assist in understanding numerical differences between different implementations and show what is needed in order to preserve the property of monotonicity in floating-point multi-term addition.
This applies to dot products and matrix multiplication operations, which use multi-term addition.
The results do not necessarily apply only to hardware as the same low level floating-point algorithms could be implemented in software, for example on devices that do not contain floating-point units.

Finally, our results indicate that monotonicity should be considered in the next iteration of the IEEE 754 standard and the new standard, P3109\footnote{\url{https://sagroups.ieee.org/p3109wgpublic/}}, of low-precision floating-point formats.
Suitable recommendations for the reduction operations may need to be provided when the preservation of monotonicity is needed.

\section{Acknowledgements}

The author is grateful to M.~Fasi for the help on the proof of Theorem~\ref{thm:ieee754-sum-monotonic} and other comments and suggestions, and N.~J.~Higham, Nicolas Brunie and three anonymous referees for comments and suggestions.

\section{Funding}

This work was supported by the Engineering and Physical Sciences Research Council (EPSRC) grant EP/P020720/1, and by the Exascale Computing Project (17-SC-20-SC), a collaborative effort of the U.S. Department of Energy Office of Science and the National Nuclear Security Administration.
The author also acknowledges the support of the School of Computing, University of Leeds.

\bibliographystyle{IEEEtran}
\bibliography{IEEEabrv,strings,njhigham,references}

%
\end{document}